\documentclass [11pt] {article}

 \usepackage{fullpage}

\usepackage{graphics}
\usepackage[dvips]{epsfig}

\usepackage{amsmath}
\usepackage{amssymb}
\usepackage{amsfonts}
\usepackage{graphicx}

\usepackage{cite}

\usepackage{algorithmic}
\usepackage[linesnumbered,ruled]{algorithm2e}

\begin{document}

\newtheorem{theorem}{Theorem}[section]
\newtheorem{lemma}{Lemma}[section]
\newtheorem{corollary}{Corollary}[section]
\newtheorem{claim}{Claim}[section]
\newtheorem{proposition}{Proposition}[section]
\newtheorem{definition}{Definition}[section]
\newtheorem{fact}{Fact}[section]
\newtheorem{example}{Example}[section]
\newtheorem{remark}{Remark}[section]
\newtheorem{case}{Case}
\newcommand{\cV}{{\cal V}}

\newcommand{\qed}{\hfill $\square$ \smallbreak}
\newenvironment{proof}{\noindent{\bf Proof:}}{\qed}

\def\thefootnote{\fnsymbol{footnote}}

\title{{\bf Universal Rendezvous of Anonymous Agents with Footprints}}\date{}

\author{Bibhuti Das\footnotemark[1]
}

\footnotetext[1]{D\'epartement d'informatique, Universit\'e du Qu\'ebec en Outaouais, Gatineau,
Qu\'ebec J8X 3X7, Canada. {\tt dasbibhuti905@gmail.com}}

\maketitle
\thispagestyle{empty}
\vspace*{0.5cm}

\begin{abstract}
Deterministic rendezvous for two anonymous mobile agents starting simultaneously from two distinct nodes of an anonymous connected graph and navigating synchronously in the graph requires that they meet at some node.
 An instance of the rendezvous problem is the underlying graph, together with two distinct nodes that are the initial positions of the agents. Such an instance is said to be {\em feasible} if there is a deterministic algorithm guaranteeing rendezvous, possibly valid only for this instance. A rendezvous algorithm is said to be {\em universal} for a class of instances if it guarantees rendezvous for all feasible instances from this class.
 
We consider the model with footprints: whenever an agent visits an unmarked node, it leaves a permanent {\em footprint} on it, and all footprints are identical. This paper aims at answering the open problem from the paper by Das and Pelc (SPAA
2026), asking whether there exists a universal rendezvous algorithm for the class of all instances in the model with footprints. We propose a universal rendezvous algorithm for the class of all instances where the underlying graph is connected (finite or countably infinite) in the model with footprints. This shows the existence of a universal algorithm for the class of all instances in the model with footprints, which is an affirmative answer to the open problem.

\vspace*{0.5cm}

\noindent
{\bf keywords:} anonymous graph, anonymous agent, rendezvous, footprints, universal algorithm

\vspace*{0.5cm}
\end{abstract}

\newpage

\section{Introduction}
The task of deterministic rendezvous for two anonymous mobile agents navigating synchronously in an anonymous network, modeled as a connected, unlabeled graph, calls for their meeting at some node, by executing the same deterministic algorithm. The initial nodes of the agents in the network are chosen by an adversary. The agents need to meet at some node in the same round to solve rendezvous. The methodological importance of rendezvous comes from the fact that it is a symmetry-breaking task equivalent to the fundamental distributed task of leader election between the agents (cf., e.g., \cite{PY}), which asks for one of the agents to become {\em leader} and the other a {\em non-leader}. In computer networks, the mobile agents may be software agents navigating in a network to consult
a distributed database whose parts are located in nodes of the network. They can also be viewed as the mobile
robots navigating in a network of corridors of a contaminated mine to collect samples of the ground or of the air. In both cases, the task of rendezvous 
may be needed to exchange data collected so far for planning future actions.

\subsection{The model and the problem}
We consider a network represented by a simple  undirected connected graph $G=(V,E)$ that can be either finite or countably infinite. The degrees of all nodes are finite but there is no global known upper bound on node degrees. Nodes are unlabeled, but ports at each individual node of degree $d$ are numbered by integers $0,1,\dots, d-1$. There is no coherence between port numbers at two extremities of an edge. Two mobile agents start at different nodes of the network, called their {\em bases} chosen by an adversary. The agents start simultaneously and navigate synchronously in the network. To achieve rendezvous, they must meet at some node in the network in the same round. 

An instance of the rendezvous problem is the underlying graph, together with two distinct nodes that are initial positions of the agents. Such an instance $(G,u,v)$ is {\em feasible} if there exists an algorithm, possibly valid only for this instance, that guarantees rendezvous for it. A rendezvous algorithm is {\em universal} for a class of instances, if it is valid for all feasible instances from this class.

In each round, an agent can either stay at the current node or move to an adjacent node by a chosen port. When an agent arrives at a node, it sees its degree and the port number by which it enters the node. Agents are identical (they have no labels) and execute the same deterministic algorithm. They have an unbounded memory; computationally they are modelled as Turing machines. A {\em rendezvous} between agents occurs when they are at the same node in the same round. The
agents can cross each other in an edge going in opposite directions but they do not notice it. This is a standard assumption for synchronous rendezvous in networks (cf. \cite{CKP,DFKP}).

In this paper, we consider the model with footprints (cf. \cite{GaP}), whenever an agent visits an unmarked node, it leaves a permanent mark, called {\em footprint}, on it, and all footprints are identical.

Let $G$ be any graph and $v$ a node in this graph. A {\em path} of length $k$ starting at $v$ is a sequence of integers $(p_1,q_1,p_2,q_2, \dots ,p_k,q_k)$, such that there exist nodes $v=v_0$, $v_1$, ..., $v_{k}$ in the graph $G$, for which $v_i$ is adjacent to $v_{i+1}$, where $i=0,1,\dots k-1$, and  edge $\{v_i,v_{i+1}\}$ has ports $p_{i+1}$ at its endpoint $v_i$ and 
$q_{i+1}$ at its endpoint $v_{i+1}$. Node $v_i$ is called the $i^{th}$ node of $\pi$.
For a path $\pi=(p_1,q_1,p_2,q_2, \dots ,p_k,q_k)$ starting at $v_0$, the {\em reverse path} $rev(\pi)$ is the path
$(q_k,p_k,q_{k-1},p_{k-1},\dots, q_1,p_1)$, starting at $v_k$. 

The {\em view} from $v$ in $G$, denoted $\cV(v,G)$, is the tree of all paths in $G$, starting at node $v$, where the rooted tree structure is defined by the prefix relation of sequences of integers. This definition is equivalent  to that from \cite{YK3}, where the notion of view was introduced. Nodes $u$ and $v$ in $G$ are called {\em symmetric}, if $\cV(u,G)=\cV(v,G)$.

In the model with footprints, Das and Pelc \cite{DP2026} presented universal rendezvous algorithms for four different classes of instances: those with symmetric bases, those with non-symmetric bases, those on trees, and those with a known upper bound on node degrees. In all the cases, the proposed universal algorithms work for graphs that are finite or countably infinite. However, they remain limited to classes of instances with a prior knowledge. Our aim is to design universal rendezvous algorithm for the class of all instances in the model with footprints without any a prior knowledge.

\subsection{Our contribution}
The main contribution of this work is to provide an affirmative answer to the following open question from \cite{DP2026}.\\

{\it Does there exist a universal rendezvous algorithm for the class of all instances in the model \\ \hspace*{8cm}with footprints?}\\

\noindent  Das and Pelc \cite{DP2026} proposed universal rendezvous algorithms in the model with footprints for two complementary classes: the class of instances with symmetric bases and the class of instances with non-symmetric bases. We design a universal rendezvous algorithm for the class of all instances in the model with footprints, where the underlying graph is an arbitrary connected graph, finite or countably infinite, without any {\em a priori} knowledge. This proves that there exists a universal algorithm for the class of all instances in the model with footprints.

\subsection{Related work}

The rendezvous problem has been investigated both in the deterministic scenario \cite{Pe} and in the randomized scenario \cite{alpern02b,alpern02a}. In many papers, the geometric scenario was considered in the context of the rendezvous task, e.g., in \cite{anderson98a}.
The extension of the rendezvous problem to several agents is usually called gathering, and was studied, e.g., 
in \cite{lim96}. 
The gathering problem for oblivious robots in the plane was studied in
\cite{FPSW},  and fault-tolerant aspects of this problem were investigated, e.g., in \cite{AP06,CP08}.

Rendezvous was studied both assuming that agents are allowed to mark visited nodes, e.g., using tokens
\cite{KKSS}, and assuming that they always mark visited nodes by leaving footprints \cite{GaP}, similarly as in the present paper. However, in most of the literature on rendezvous, the study considers a model forbidding any marking of nodes. The symmetry between the agents is often broken by assigning different labels to each of them \cite{DFKP}. It is then assumed that each agent knows its own label but not the other agent's label. In \cite{BP}, the authors designed the first rendezvous algorithm that works on arbitrary (even countably infinite) graphs for labeled agents that navigate synchronously, without leaving marks.

Deterministic rendezvous of anonymous agents in finite anonymous graphs was previously studied in \cite{CKP,DP1}. In the paper \cite{CKP}, the authors proposed a rendezvous algorithm that works for all non-symmetric initial positions in arbitrary finite graphs with a known bound on their size, using memory logarithmic in this bound. In the work \cite{DP1}, the authors studied the gathering problem of anonymous agents and characterized initial positions that allow gathering. The authors of \cite{GaP,DP2026} considered rendezvous in the model with footprints. A large volume of works investigated asynchronous rendezvous in the plane \cite{CFPS,FPSW} and in networks 
\cite{BCGIL,CLP,DGKKP,DPV}.

\section{Universal rendezvous with footprints}

In this section, we design a universal rendezvous algorithm for the class of all instances in the model with footprints and prove its correctness. 

\subsection{ Preliminaries}
In the model with footprints, whenever an agent visits an unmarked node, it leaves a permanent footprint, and all footprints are identical. 
We will use the following notion from \cite{DP2026}. Consider an instance $(G,u,v)$ and a path $\pi$ in $G$. A {\em parallel run} $[G,u,v,\pi]$ is the following procedure.
One agent starts at $u$, and the other agent starts at $v$.
Let $(u_0=u,u_1,\dots, u_k)$ and $(v_0=v,v_1,\dots, v_k)$ be the sequences of nodes corresponding to path $\pi$ starting at nodes $u$ and $v$, respectively.
Each agent follows path $\pi$ in rounds $1,2,\dots, k$. The {\em footprint status} for a parallel run $[G,u,v,\pi]$ of a node $u_i$ (respectively $v_i$) is 1, if this node is marked at its first visit by the agent starting at $u$ (respectively starting at $v$) during the forward exploration of the path, and it is 0 otherwise.

The following proposition from \cite{DP2026} gives a necessary condition on the feasibility of an instance in the model with footprints. 

\begin{proposition}\label{condition}
If the instance $(G,u,v)$ is feasible then
\begin{itemize}
\item
 nodes $u$ and $v$ are non-symmetric, or
 \item
 nodes $u$ and $v$ are symmetric and there exists a path $\pi$ of length $k$ in $G$,
 such that, for the parallel run $[G,u,v,\pi]$, there exists an  $1\leq i\leq k$, for which the footprint status of nodes $u_i$ and $v_i$ is different.
 \end{itemize}
\end{proposition}

Intuitively, the above condition says that either the views from the agents' bases are different, or they are equal, but the ``input on the ground'' (i.e., the footprint status) perceived by the agents along some path must differ.

Next, we recall the transformation
of binary strings defined in \cite{BP}. Consider a binary string $L=(l_1 l_2\ldots l_s)$ with $l_1=1$. Following \cite{BP,DP2026}, we define the following transformations of $L$.

First, we construct the binary string $T_1(L)$ from $L$ as follows. Each bit 1 is replaced by 10, and each bit 0 is replaced by 01. Additionally, bits
11 are added at the end. The obtained string $T_1(L)$ has length $2s + 2$ and has the following property.\\
{\bf Property 1.} If  we
start with two different strings $L_1\neq L_2$ then none of the transformed strings $T_1(L_1)$ and $T_1(L_2)$ can be a prefix of the other (cf. \cite{DP2026}).

Next, we construct the string $T_2(L)$ of length $4s+4$ from $T_1(L)$ by replacing each bit 1 by 10 and each bit 0 by 01. Notice that since two strings $L_1$ and $L_2$
may have different lengths, the same is true for the transformed
strings $T_2(L_1)$ and $T_2(L_2)$. However, they have the following property.\\
{\bf Property 2.}
 If $L_1\neq L_2$, then there exists an
index $j$, such that the $j^{th}$ bit of $T_2(L_1)$ is 1 and the $j^{th}$ bit of $T_2(L_2)$ is 0 (cf. \cite{DP2026}).  (Hence, not only $T_2(L_1)$ and $T_2(L_2)$ differ in some bit, but it is guaranteed that in some
position the bits are 1 and 0, and in some other position the bits are 0 and 1). 

Finally, we construct $T_3(L)$ from $T_2(L)$ by adding the prefix 011 to it. The transformed string $T_3(L)$ of length $4s+7$ preserves the property of being prefix-free and guarantees that if $L_1\neq L_2$ then in some
position the bits are 1 and 0 and in some other position the bits are 0 and 1. The transformed string $T_3(L)$ has the following additional property.\\
{\bf Property 3.}
The segment corresponding to bits with indices 2,3,4 in $T_3(L)$ is the only segment of three consecutive bits 1 in $T_3(L)$ (cf. \cite{DP2026}). 

For example, consider $L=10$. We have $T_1(L)=100111$, $T_2(L)=100101101010$, and $T_3(L)=011100101101010$.

\subsection{High-level idea of the algorithm}
The high-level idea of $AlgorithmFootprint$, which guarantees rendezvous in any connected graph in the model with footprints, is as follows. From Proposition~\ref{condition}, it follows that a feasible instance can have either non-symmetric bases or symmetric bases. If the views from the bases are different, there exists some path that differs when starting from each of the bases. This path, converted to a binary string and transformed, is used to break the symmetry in the agents' actions, enabling rendezvous. On the other hand, if the views from the bases are the same, then the footprint status perceived by the agents following some path must be different. In other words, by following some path, one agent visits an unmarked node, while the other agent visits the corresponding node by following the same path that is already marked. Based on the footprint status, the agents break symmetry between their actions to meet. 

Since the agents don't know when such a difference will be found, the algorithm proceeds by increasing levels. At each level $k$, an agent explores all paths of length $k$, and constructs the {\it temporary label} $B_k$ (based on port numbers), and the {\it footprint label} $F_k$ (based on footprint status). At each level $k\geq 1$, the agents first process the bits of the {\it temporary label} $B_k$. If there is no rendezvous at the end of processing all bits of $B_k$, then the agent processes the bits of the {\it footprint label} $F_k$. In other words, at each level, the agents check the possibility of non-symmetric bases at first, and then check the possibility of symmetric bases. If there is still no rendezvous, the agent then proceeds to the next level. 

The agents process a bit by being active when the bit is 1 and passive when the bit is 0. An active agent traverses some path and returns to the base, while a passive agent waits at the base. At some point, the path that exhibits the difference (based on port numbers or footprint status) is found and is sufficiently long to reach the base of the passive agent. This causes the agents to process different bits, 1 and 0, synchronously, where the path traversed by one agent forces it to meet the agent that processes bit 0, who is waiting at its base.

\subsection{Detailed description of the algorithm}
For any $k\geq 1$, we define the {\em temporary label} $B_k$  of an agent $A$ at level $k$ as follows. Consider all $k$-length paths starting at the base of $A$. Each such path is a sequence of length $2k$
of non-negative integers that are port numbers. Assume that $\pi_1^k,\pi_2^k,\ldots,\pi_t^k$ is the list of all paths of length $k$ starting at the base of $A$, in lexicographically increasing order. Let $\pi_1^k=(p_1,q_1,p_2,q_2, \dots ,p_k,q_k)$.  Let $b_1^k$ be the concatenation of binary representations of integers $p_1,q_1,p_2,q_2, \dots ,p_k,q_k$. Similarly, we define $b_i^k$, for all $i=1,2,\dots,t$.
For $1\leq i\leq t$, let $c_i^k=T_3(b_i^k)$. Finally, define the temporary label
 \[
B_k = \underbrace{c_1^kc_1^k\ldots c_1^k}_{\text{$t$ times}} \underbrace{c_2^k c_2^k\ldots c_2^k}_{\text{$t$ times}}\ldots \underbrace{c_t^k c_t^k\ldots c_t^k}_{\text{$t$ times}}\]
Next, for any $k\geq 1$, we define the {\em footprint label} $F_k$ of an agent $A$ at level $k$. For any $k\geq 1$ and $1\leq i\leq t$, let $v_i^k$ denote the $k^{th}$ node of $\pi_i^k$. Define $V_k=s_1s_2\ldots s_t$ where for $1\leq i \leq t$,\[ s_i=\begin{cases} 
1 & \textrm{ if} \;v_i^k \textrm { is a marked node at the first visit following path}\;\pi_i^k \;\textrm{by the agent}  \\
0 &  \textrm{ if} \; v_i^k \textrm { is an unmarked node at the first visit following path}\;\pi_i^k\;\textrm{by the agent} \\
\end{cases}
\]

\noindent By the first visit, we mean the visit during the forward exploration
of the path $\pi_i^k$ by the agent. Notice that $V_k$ is a binary string whose first bit can be 0. For each $k\geq 1$, let $W_k=V_1V_2\ldots V_k$ and $f_k=T_2(W_k)$. Finally, define the footprint label \[F_k=
 \underbrace{f_kf_k\ldots f_k}_{\text{$t$ times}}\]
The {\em tape} of an agent $A$ is defined as the infinite binary string $tape(A)$ as the concatenation of temporary labels $B_k$, and footprint labels $F_k$ for all $k\geq 1$ in the following way $tape(A)=(B_1 F_1 B_2 F_2\ldots B_k F_k\ldots )$. The position of a bit in the tape is called its {\em index}. For each $k\geq 1$, the $k^{th}$ block of $tape(A)$ consists of $B_k$ and $F_k$.

Let $n_k$ denote the number of $k$-length paths starting at the base of agent $A$. For each $k\geq 1$, $B_k$ is divided into $n_k$ segments. For $1\leq j\leq n_k$, the $j^{th}$ segment of $B_k$ consists of $n_k$ copies of $c_j^k$. For $1\leq p \leq n_k $, the $p^{th}$ copy of $c_j^k$ in the $j^{th}$ segment of $B_k$ will be related to the traversal of the path $\pi_p^k$ by agent $A$, i.e., the $p^{th}$ path of length $k$ in lexicographic order. On the other hand, for each $k\geq 1$, $F_k$ consists of only one segment that has $n_k$ copies of $f_k$. For $1\leq p \leq n_k $, the $p^{th}$ copy of $f_k$ in $F_k$ will be related to the traversal of the path $\pi_p^k$ by agent $A$, i.e., the $p^{th}$ path of length $k$ in lexicographic order.

The algorithm is executed by an agent $A$ with
tape $tape(A)$, starting at its base $b_A$. We assign rapidly increasing time periods to process consecutive bits $l_i$ of $tape(A)$. During the execution of the algorithm, agent $A$ constructs $tape(A)$ step by step. By the string transformation, for each $k$, the first two bits of a temporary label $B_k$ are 01. By processing the first bit 1, agent $A$ traverses the lexicographically smallest path of length $k$ and constructs $c_1^k$. As the algorithm proceeds, agent $A$ explores the $p^{th}$ lexicographically ordered path of length $k$ by processing the bits of the $p^{th}$ copy of $c_1^k$. Once agent $A$ explores all paths of length $k$, it can construct the temporary label $B_k$ and footprint label $F_k$ of $tape(A)$.

The period reserved for processing bit $l_i$ lasts $3i^3$ rounds. The processing bit $l_i$ can be in the the $p^{th}$ copy of $c_j^k$ in the $j^{th}$ segment of $B_k$, or in the 
the $p^{th}$ copy of $f_k$ in $F_k$, where $k\geq 1,\; 1\leq j\leq n_k,$ and $1\leq p \leq n_k$. Then the processing of bit $l_i$ concerns the path $\pi_p^k$ starting
at $b_A$, in the following way:

\begin{itemize}
\item $l_i=1$. For each $1\leq q\leq k$, the agent goes to the $q^{th}$ node of the path $\pi_p^k$, waits for $3i^2-2q$ rounds, and returns to the base. Next, the agent waits for $3i^3- k\cdot 3i^2$ rounds at the base.
\item $l_i=0$. The agent waits for $3i^3$ rounds at the base. 
\end{itemize}
Note that the agent $A$ starts and ends processing each bit of $tape(A)$ at its base $b_A$. The algorithm is interrupted as soon as the agents meet. The pseudocode of the algorithm, executed by agent $A$, is presented on the next page.

\begin{algorithm}[h]
\footnotesize
Initialize $i=1$, and $B_0=\varepsilon$; \tcp{$\varepsilon$ is the empty string}
\For{$k=1,2\ldots$}
{$tape(A)=B_1 F_1 B_2 F_2\ldots B_{k-1}F_{k-1}01$ and $x=|tape(A)|$\;
\While{$i\leq x$}
{Let $l_i$ be a bit in the $p^{th}$ copy of $c_j^k$ in the $j^{th}$ segment of $B_k$ or $l_i$ be a bit in the $p^{th}$ copy of $f_k$ in $F_k$\;
\uIf{$l_i=1$ } 
    {\For{$q=1$ to $k$} 
    {Go to the $q^{th}$ node of the path $\pi_p^k$, wait for $3i^2-2q$ rounds, and return to the base}
    Wait for $3i^3- k \cdot 3i^2$ rounds at the base\;}
 \Else
    { 
     Wait for $3i^3$ rounds at the base\;
    }
 \uIf{$j=1$, $p=1$, and $l_i$ is the first bit 1 in $B_k$ and $\pi_1^k$ is not the lexicographically largest path of length $k$}{$tape(A)=B_1 F_1B_2 F_2\ldots B_{k-1}F_{k-1}c_1^kc_1^k$ and $x=|tape(A)|$\; } 
 \uIf{$j=1$, $p=1$, and $l_i$ is the first bit 1 in $B_k$ and $\pi_1^k$ is the lexicographically largest path of length $k$}{$tape(A)=B_1F_1B_2F_2\ldots B_{k-1}F_{k-1}B_k F_k$ and $x=|tape(A)|$\; } 
 \uIf{$j=1$, $p>1$, and $\pi_p^k$ is not the lexicographically largest path of length $k$ and $l_i$ is the first bit 1 in the $p^{th}$ copy of $c_1^k$} {$tape(A)=tape(A).c_1^k$ and $x=|tape(A)|$\;}
    \If{$j=1$, $p>1$, and $\pi_p^k$ is the lexicographically largest path of length $k$ and $l_i$ is the first bit 1 in the $p^{th}$ copy of $c_1^k$} {$tape(A)=B_1F_1B_2F_2\ldots B_{k-1}F_{k-1}B_k F_k$ and $x=|tape(A)|$\;}$i=i+1$\; }}
         
\caption{$AlgorithmFootprint$}
\label{algo1}
\end{algorithm}

\subsection{Correctness}
In this section, we prove the correctness of $AlgorithmFootprint$. From Proposition~\ref{condition}, it follows that a feasible instance can have either non-symmetric bases or symmetric bases. First, consider the case when the agents start from non-symmetric bases $b_A$ and $b_B$. They must have different paths from some length on. Let $k$ be the smallest integer for which there are different paths of length $k$ starting from $b_A$ and $b_B$. 
For
an instance of the rendezvous problem, where agents $A$ and $B$ start at bases $b_A$ and $b_B$, respectively, where $D$ is the distance between their bases, we define the {\it critical path} for $A$ as follows. First, consider the case when $k\leq D$. Let $\pi$ denote the lexicographically smallest among shortest paths from $b_A$ to $b_B$. In this case, we call $\pi$ the critical path of agent $A$. The length of this path is $D$.
Next, consider the case when $k>D$. Let $\pi$ denote the lexicographically smallest $k$-length path starting at $b_A$, and passing through $b_B$. In this case, we call $\pi$ the critical path of agent $A$.
The length of this path is $k$. 

Next, consider the case when the agents start from symmetric bases $b_A$ and $b_B$. Since $\cV(b_A,G)=\cV(b_B,G)$, by Proposition~\ref{condition} there exists a level $y\geq 1$ such that the footprint label $F_y$ is different for agent $A$ and for agent $B$. By the definition of footprint label, for all $k\geq y$, $F_k$ is different for agent $A$ and for agent $B$. Let $\pi$ denote the lexicographically smallest $k$-length path starting at $b_A$, and passing through $b_B$ such that $k\geq y$. We call $\pi$ the critical path of agent $A$. 

Let us assume that the critical path $\pi$ of length $\alpha$ of agent $A$ is the $j^{th}$ lexicographically smallest $\alpha$-length path starting at $b_A$, i.e., $\pi_j^{\alpha}=\pi$. The segment of $tape(A)$ corresponding to the $j^{th}$ copy of $c_j^{\alpha}$ in the $j^{th}$ segment of $B_k$ or the $j^{th}$ copy of $f_k$ in $F_k$ is called the {\it critical segment} of agent $A$. The agents start simultaneously from their bases. However, the agents may start executing their respective critical segments in different rounds based on their views. 

\begin{remark}\normalfont
To show that the formulation of the algorithm is correct, we need to prove that if $l_i$ is the $i^{th}$ bit of $tape(A)$ located in $B_k$ or $F_k$, then $\forall 1\leq j\leq k,\; 3i^2\geq 2j$. Indeed, we have $3i^2\geq 2i\geq 2k\geq  2j$, for all $1\leq j\leq k$.
\end{remark}

The following lemma proves the correctness of $AlgorithmFootprint$ when the agents start from non-symmetric bases $b_A$ and $b_B$. 

\begin{lemma}\label{th:nonsym} 
Suppose that the agents start executing $AlgorithmFootprint$ simultaneously, and their bases are non-symmetric. Then the agents meet by the end of the execution of the critical segment by the agent that starts executing its critical segment earlier.
\end{lemma}
\begin{proof}
Let $A$ be the agent that starts the execution of its critical segment earlier, and let $B$ be the other agent. If agents start executing their critical segments simultaneously, then the choice of $A$ is arbitrary.
Let $r_1$ be the round in which $A$ starts executing its critical segment. Let  $m-1$ be the index of the first bit in this segment. Hence, agent $A$ starts the execution of the second
bit of its critical segment  (which is bit 1) in round $r_2 = r_1 + 3\cdot (m-1)^3+1$. Since the agents start simultaneously, and the time required for processing the $i^{th}$ bit of their respective tapes is equal to $3i^3$ rounds for both the agents, agent $B$ also starts processing the $m^{th}$ bit of $tape(B)$ in round $r_2$. 

Consider the time interval $[r_2, r_2 + (D-1)\cdot 3m^2+D]$. During this
time interval, agent $A$ traverses its critical path $\pi$ (of length $\alpha$), and reaches the $D^{th}$ node on $\pi$ by processing the $m^{th}$ bit of $tape(A)$. If the $m^{th}$ bit of $tape(B)$ is 0, then agent $A$ meets agent $B$ at the base $b_B$ during the execution of the $m^{th}$ bit of  $tape(B)$ by $B$, because agent $B$ will be waiting at $b_B$, while agent $A$ reaches $b_B$.
Otherwise, the $m^{th}$ bit of $tape(B)$ is 1, and we have the following three subcases.
\begin{enumerate}
\item the $m^{th}$ bit of agent $B$ is not the first bit 1 of some $p^{th}$ copy of $c_j^k$ in the $j^{th}$ segment of $B_k$. Then by the definition of $c_j^k$, at least one among the $(m+ 1)^{th}$ and $(m+2)^{th}$ bits of $tape(B)$ is 0, because of Property 3 of string transformation. 
Since the $m^{th}$, the $(m+ 1)^{th}$ and the $(m+2)^{th}$ bits of $tape(A)$ are 1, the agents meet at the base $b_B$ of $B$ during the execution of the first bit 0 of $tape(B)$ following the $m^{th}$ bit of this tape. During the execution of this bit, agent $B$ waits at its base $b_B$ and agent $A$ reaches $b_B$.

\item the $m^{th}$ bit of agent $B$ is the first bit 1 of some $p^{th}$ copy of $c_j^k$ in the $j^{th}$ segment of $B_k$. 
In view of Properties 1 and 2 of string transformation, for some index $s$, the $s^{th}$ bit of $tape(A)$ is 1, the $s^{th}$ bit of $tape(B)$ is 0. Hence, the agents meet at the base $b_B$ of $B$, during the executions of the $s^{th}$ bits of their tapes.
Indeed, during the execution of this bit, agent $B$ waits at its base $b_B$ and agent $A$ reaches $b_B$.

\item the $m^{th}$ bit of agent $B$ is in some $p^{th}$ copy of $f_k$ in $F_k$. By the definition of $f_k$, at least one among the $(m+ 1)^{th}$ and $(m+2)^{th}$ bits in $F_k$ of $tape(B)$ is 0, because there are no three consecutive bits 1 in $f_k$. During the execution of this bit, agent $B$ waits at its base $b_B$ and agent $A$ reaches $b_B$.
\end{enumerate}

\vspace*{-0.75cm}
\end{proof}

The following lemma proves the correctness of $AlgorithmFootprint$ when the agents start from symmetric bases $b_A$ and $b_B$. 

\begin{lemma}\label{th:sym} 
Suppose that the agents start executing $AlgorithmFootprint$ simultaneously, and their bases are symmetric. Then the agents meet by the end of the execution of the critical segment by the agent that starts executing its critical segment earlier.
\end{lemma}
\begin{proof}
Due to the symmetry of the bases and the simultaneous start
of the agents, the processing of the bits by both agents proceeds simultaneously for all levels $k\geq 1$. Let $A$ be the agent that starts the execution of its critical segment earlier, and let $B$ be the other agent. If agents start executing their critical segments simultaneously, then the choice of $A$ is arbitrary. Let the critical segment of agent $A$ correspond to the $p^{th}$ copy of $f_k$ in $F_k$ for some $k\geq 1$. By the definition of footprint label, in view of Properties 1 and 2 of string transformation, and by Proposition~\ref{condition}, for some index $j$, the $j^{th}$ bit in the critical segment of $A$ is 1, and the $j^{th}$ bit in $F_k$ of $tape(B)$ is 0. Hence, the agents meet at the base $b_B$ of $B$, during the executions of the $j^{th}$ bits of their tapes.
Indeed, during the execution of this bit, agent $B$ waits at its base $b_B$ and agent $A$ reaches $b_B$.
\end{proof}

\begin{theorem}\label{univ}
$AlgorithmFootprint$ is a universal rendezvous algorithm for the class of all instances, where the underlying graph is any connected graph, finite or countably infinite, in the model with footprints.

\end{theorem}
\begin{proof}
Consider a feasible instance $(G,u,v)$ in the model with footprints. Suppose the agents start simultaneously. From Proposition~\ref{condition}, it follows that a feasible instance can have either non-symmetric bases or symmetric bases. If the bases are non-symmetric, then rendezvous is guaranteed by Lemma~\ref{th:nonsym}. In case bases are symmetric, then rendezvous is guaranteed by Lemma~\ref{th:sym}. This proves the theorem.
\end{proof}

The following corollary implied by Theorem~\ref{univ} is the main result of this paper that answers the open question from the paper \cite{DP2026}.

\begin{corollary}
There exists a universal rendezvous algorithm for the class of all instances, where the underlying graph is any connected graph, finite or countably infinite, in the model with footprints.
\end{corollary}
\section{Conclusion}

In the model with footprints, we proposed a universal rendezvous algorithm for the class of all instances where the underlying graph is any connected graph that is finite or countably infinite. This resolves the open problem of whether a universal rendezvous algorithm exists for the class of all instances in the model with footprints.

\end{document}